\documentclass[preprint,12pt,authoryear]{elsarticle}

\usepackage{amssymb}
\usepackage{amsthm}
\usepackage{amsmath}

\usepackage[T1]{fontenc}
\usepackage{lmodern}
\usepackage{microtype}
\usepackage[colorlinks]{hyperref}

\newcommand{\nonnegreal}{{\mathbb R}_+}
\newcommand{\posreal}{{\mathbb R}_{++}}
\newcommand{\rulename}{\operatorname}
\newcommand{\problems}{\mathcal{C}}

\makeatletter
\newcommand{\aru}{\@ifstar\@saru\@aru}
\newcommand{\@saru}[1]{\mathop{\left(\exists #1\right)}}
\newcommand{\@aru}[1]{\mathop{\exists #1}}
\makeatother
\newcommand{\zenbu}[1]{\mathop{\forall #1}}
\newcommand{\inv}{{-1}}

\newtheorem{thm}{Theorem}[section]
\newtheorem{cor}[thm]{Corollary}
\newtheorem{lem}[thm]{Lemma}
\newtheorem{prop}[thm]{Proposition}
\newtheorem*{prop*}{Proposition}
\newtheorem*{fact*}{Fact}
\newtheorem*{claim*}{Claim}
\newtheorem{fact}[thm]{Fact}

\theoremstyle{definition}
\newtheorem{defi}[thm]{Definition}
\theoremstyle{remark}
\newtheorem*{rem*}{Remark}

\journal{European Journal of Operational Research}

\begin{document}

\begin{frontmatter}



\title{A Certain Notion of Strategy Freedom \mbox{under Retail Competition} in Claims Problems}


\author[label1]{Kentar\^{o} Yamamoto}
\ead{ykentaro@g.ecc.u-tokyo.ac.jp}
\affiliation[label1]{organization={Research Center for Advanced Science and Technology, \mbox{The University of Tokyo}},
            addressline={4--6--1~Komaba}, 
            city={Meguro},
            postcode={153--8904}, 
            state={Tokyo},
            country={Japan}}

\begin{abstract}
  A new axiom for rules for claims problems is introduced.
  It strengthens a condition studied in supply chain literature,
  which forces rules to
  disincentivize order inflation under capacity allocation and retail competition.
  The relevance of the axiom is further demonstrated by
  one of the main results of the present article:
  it
  characterizes the weighted constrained equal awards rule
  together with known natural axioms.
\end{abstract}




\end{frontmatter}


\section{Introduction}
Formal studies on claims problems, also known as bankruptcy problems,
and rules applicable to them
are motivated by needs to distribute some resource
claimed by several agents
the sum of whose claims adds up to more than the said resource.
The vast body of research that began with \cite{ONEILL1982345} has been summarized in the
monograph by \cite{Thomson2019}.

An object of study appearing in the literature that is important for our purpose
is the
weighted constrained equal awards rule.
This was introduced by \cite{lee94:_simpl_gener_const_equal_award_rule_charac}
and characterized in the broader context of weighted claims problems
by \cite{https://doi.org/10.1111/1468-0262.00126}.

The present article introduces a new axiom on weighted rules for claims problems.
The axiom is inspired by practical issues arising in the supply chain literature.
\cite{Liu2012EquilibriumAO}
considers a situation
where a single supplier receives orders from many retailers
who engage in some Cournot competition of complete information.
When the sum of the retailers' orders exceeds his stock,
the supplier applies a rule to the associated claims problem
whose claimants are the retailers and whose claims are their orders.
\cite{lee94:_simpl_gener_const_equal_award_rule_charac} pointed out that
depending on the rule chosen,
a retailer may strategically inflate her order, that is,
making an order larger than one expected solely from the Nash equilibrium
of the said Cournot competition.
\cite{Cho2014} discussed how to choose a rule that eliminate such strategic
inflation of orders.

While \cite{Liu2012EquilibriumAO}
and \cite{Cho2014}
assumed that the retailers engage in competitions of a certain form
and studied elimination of inflation of orders
under this assumption,
we instead ask for rules eliminating inflation of orders
\emph{regardless} of what competition the retailers are in.
This leads to a new axiom on weighted rules on claims problems,
which we  call strategy-freedom in the present article.
Extracting this axiom from the supply chain literature
allows us to apply axiomatic methods to the supply chain issues.
We reveal what rules, if any, satisfy strategy-freedom;
the relationship between strategy-freedom and other axioms on rules;
and the characterization of weighted constrained equal awards rule
in terms of strategy-freedom.

The article is organized as follows.
In \S~\ref{sec:preliminaries},
we review standard concepts in claims problem literature
while fixing the terminology and the notation.
We dedicate \S~\ref{sec:new-axioms}
to the derivation of our new axiom, strategy-freedom,
from the works by \cite{Liu2012EquilibriumAO} and \cite{Cho2014}.
We review their model and detail the relationship between their motivation
and our axiom.
We also note another axiom they implicitly studied,
whose special case can be found in the existing claims problem literature.
In \S~\ref{sec:main-results},
we engage in axiomatic study of rules in terms of the newly-found axioms.
We first study the relationship between the two axioms inspired by
\cite{Cho2014} and clarify certain points made by them.
We then examine the details of the axiom of strategy-freedom
and reveal its interesting behavior across different weights.
In light of this, we devise
an alternate version of strategy-freedom, continuous strategy-freedom,
and characterize it in terms of the original strategy-freedom and continuity.
Lastly, we prove the characterization of the weighted constrained equal awards
rule in terms of strategy-freedom and several other known axioms.
We conclude the article by stating \emph{inter alia} an open problem
in \S~\ref{sec:concluding-remarks}.
The \ref{sec:rule-strategy-freew} has more technical results on the fine details of
strategy-freedom.

\section{Preliminaries}\label{sec:preliminaries}
Our framework will involve weighted versions of familiar rules and axioms;
however, weights will not be part of specifications of problems,
and instead we shall discuss families of rules and families of axioms,
both indexed by weights.
Therefore, a \emph{(claims) problem} is any vector
$(E, c) \in \nonnegreal \times \nonnegreal^{N}$,
where $N$ is the set of \emph{claimants}.
A problem~$(E, c)$ is \emph{wellformed} if $\sum c \ge E$,
where for a vector $v \in \nonnegreal^N$, we let $\sum v = \sum_{i \in N} v_i$.%
\footnote{%
  In the literature, it is customary to define problems to be what we call wellformed problems.
  We instead allow any vector of the correct index set
  with nonnegative components to be a problem
  and demand (\ref{eq:convenience}) of rules.
  This is so that the description of the game $G^{z, E}(\pi)$ may be simpler
  in \S~\ref{sec:motivating-model};
  it does not alter the theory beyond the superficial difference.}
We write $\problems^N$ for the set of \emph{wellformed} problems.
A \emph{weight} is  any member of the $(|N|-1)$-dimensional unit 
simplex $\Delta := \Delta^{|N|-1}$ in $\nonnegreal^N$
(our convention is that this refers to what some call the \emph{interior} of such a simplex).
A \emph{rule} is a function $z^w$ mapping a problem
to a vector in $\nonnegreal^N$ satisfying
\begin{equation}\label{eq:balance}
  \sum c > E \implies \sum z^w(E, c) = E,
\end{equation}
\begin{equation}
  \label{eq:convenience}
  \sum c \le E \implies z^w(E, c) = c,
\end{equation}
and
\begin{equation}
z^w(E, c) \leqq c.\label{eq:not-exceeding-claims}
\end{equation}
Finally, a \emph{weighted rule} is a family $(z^w)_{w \in \Delta}$
of rules indexed by weights,
and a \emph{weighted axiom} or a \emph{weighted property}
is a family of axioms indexed by weights.

Let us recall known rules and axioms that are relevant in our work.
Our terminology is based on \cite{Thomson2019}.
The following weighted rule
was introduced essentially by \cite{Kalai1977} and studied
by \cite{lee94:_simpl_gener_const_equal_award_rule_charac}.
\begin{defi}
  The
  \emph{constrained equal awards rule with a weight} $w$, $\rulename{CEA}^w$,
  is defined by
  \[
    \rulename{CEA}^w(E, c) = (c_i \wedge \lambda w_i)_{i \in N}
  \]
  where $\lambda$ is chosen for (\ref{eq:balance}) to hold.
\end{defi}

The following axioms are standard.
\begin{defi}
  \begin{enumerate}
  \item A rule $z$ is \emph{homogeneous} if for every $(E, c) \in \problems^N$
    and $\lambda > 0$, we have $z(\lambda E, \lambda c) = \lambda z(E, c)$.
  \item A rule $z$ is \emph{claims monotonic}
    if $z_i(E, (c_i', c_{-i})) \ge z_i(E, c)$
    for every $(E, c) \in \problems^N$, $i \in N$,
    and $c'_i > c_i$.
  \item A rule $z$ is \emph{anonymous}
    if
    $z_{\sigma(i)}(E, c) = z_i(E, c^\sigma)$,
    where $c^\sigma = (c_{\sigma(i)})_{i \in N}$,
    for every $i \in N$, $(E, c) \in \problems^N$,
    and every permutation $\sigma$ of $N$, i.e., every bijection on $N$.
  \item A rule $z$ is \emph{uniformly claims continuous}
    if $z(E, \cdot)$ is uniformly continuous as a function
    $\nonnegreal^N \to \nonnegreal^N$ for each $E > 0$.
  \end{enumerate}
\end{defi}

The following class of rules, which is broad enough to contain both the proportional rule
and $\rulename{CEA}^w$, was introduced and axiomatized by \mbox{\cite{Harless2016}}.
\begin{defi}
  \emph{Separable directional rules} or, as originally called, \emph{weighted proportional rules} are rules of the form $P^u$,
  where $u = (u_i)_{i \in N}$ is a family of functions
  $u_i \colon \nonnegreal \to \nonnegreal$
  that are strictly positive on $\posreal$ with $u_i(0) = 0$,
  and
  \[
    P^u(E, c) = (c_i \wedge \lambda u_i(c_i))_{i \in N},
  \]
  where $\lambda$ is chosen to satisfy (\ref{eq:balance}),
  for every $(E, c) \in \problems^N$.
\end{defi}


\section{New axioms}\label{sec:new-axioms}
\subsection{Motivating model}\label{sec:motivating-model}
We review the model studied by \cite{Liu2012EquilibriumAO} and \cite{Cho2014},
albeit in an abstract fashion.
The model involves a supplier of homogeneous goods
and a number of retailers.
We denote the set of retailers also by $N$,
as they will play the role of claimants in the body of the present article.
The supplier has $E$ units of the goods.
Each of the retailers $i$ puts an order to the retailer,
requesting $c_i$ units of goods ($i \in N$).
The retailer has a function $z^w$ that maps a vector $(E, c)$,
where $c = (c_i)_{i \in N}$, to a vector $z^w(E, c) \in \nonnegreal ^N$,
where $z_i^w(E, c)$ is the amount of goods that the retailer~$i$ receives
from the supplier.
This function is common knowledge among the retailers and the supplier
prior to the orders.
Let $c' = z^w(E, c)$.
We posit that the retailer~$i$ obtains the payoff $\pi_i(c')$.
In the works of \cite{Liu2012EquilibriumAO} and \cite{Cho2014}, $\pi := (\pi_i)_{i \in N}$ is the payoff functions of
some Cournot competition with a linear inverse demand function.
Now, for a moment, let us forget that the supplier had a limited supply of $E$
and assume $z^w(E, c) = c$.
Then the retailers would play the complete-information game~$G_0(\pi)$ of normal form determined by the payoff function $\pi$,
which  we assume, for the sake of simplicity, has a unique Nash equilibrium $c^* \in \nonnegreal^N$.
On the other hand, in view of the limited supply,
it is plausible to demand that $z^w$ be a rule
as defined in \S~\ref{sec:preliminaries}.
Thus, in general, $z^w(E, c) \neq c$,
and the retailers instead play the game $G^{z^w, E}(\pi)$ determined by $(\pi_i \circ z^{w, E})_{i\in N}$,
where $z^{w, E}(c) = z^w(E, c)$.
The Nash equilibrium $c^*$ of the original game $G_0(\pi)$ may cease to be
a Nash equilibrium of $G^{z^w, E}(\pi)$.
Seeing this as undesirable,
\cite{Liu2012EquilibriumAO} and \cite{Cho2014} studied conditions on $z^w$ for $c^*$ to remain a Nash equilibrium of
$G^{z^w, E}(\pi)$
(again, where $\pi$ is determined by Cournot competitions with linear demands,
unlike our approach).

\subsection{New axioms}
The aforementioned model motivates the following axioms on rules
for claims problems.
We should like our axiom to state
\begin{quote}
the Nash equilibrium $c^*(\pi)$ of $G_0(\pi)$ is still a Nash equilibrium of $G^{z^w, E}(\pi)$.
\end{quote}
Let us call this condition $P(\pi, E, z^w)$.
In order to turn this into an axiom on rules,
we need to specify $\pi$ and $E$.
\cite{Liu2012EquilibriumAO} and \cite{Cho2014} chooses specific $\pi$, namely, those coming from Cournot competitions.%
\footnote{\cite{Cho2014} actually establishes properties of \emph{pairs} of rules and $E$,
  not rules \emph{per se}.}
We, on the other hand, propose that our axiom demand
\begin{quote}
$P(\pi, E, z^w)$
for \emph{every} payoff function $\pi \colon \nonnegreal^N \to \nonnegreal^N$ with the unique Nash equilibrium $c^*$
and \emph{every} $E \ge \sum c^*$.
\end{quote}
We call this condition $Q(c^*, z^w)$.
Our proposal is based on the idea that our axiom should be a criterion
by which the supplier in the original model can choose a good rule.
Since it is more plausible that
the supplier is not cognizant of the details of the competition that
the retailers are in,
the supplier may want to choose a rule that do not alter, and hence obfuscate,
the Nash equilibrium of $G_0(\pi)$ \emph{no matter what $\pi$ is}
(we still assume that the supplier somehow knows the Nash equilibrium $c^*$).

As for our treatment of $E$ in the condition~$Q(\cdot, \cdot)$,
the demand that $E \ge \sum c^*$ is not present in \cite{Liu2012EquilibriumAO} or \cite{Cho2014}.%
\footnote{%
  A more precise statement would be that
  the question of whether or not to allow $E < \sum c^*$ does not arise
  for \cite{Cho2014},
  who in reality study pairs $(E, z^w)$.}
We chose to impose $E \ge \sum c^*$ because, otherwise,
$z^w$ could not possibly allocate $c^*$,
and whether or not $c^*$ remains a Nash equilibrium in $G^{z^w, E}(\pi)$
would depend on $\pi$.
The universal quantifier for $E$ is justified
by the same reasoning as was with payoff functions:
in absence of prior knowledge on $E$, the supplier may want to
demand of the rule that it not disrupt the Nash equilibrium no matter what $E$ is.
Our choice is also in line with the existing literature on
strategy-freedom, which is usually a universally quantified condition.

The universal quantifier for $\pi$ helps simplify $Q(\cdot, \cdot)$.
Suppose $E \ge \sum c^*$.
The vector $c^*$ is the Nash equilibrium of $\pi \circ z^{w, E}$
for all $\pi$ with the unique Nash equilibrium~$c^*$
if and only if
$z^w_{-i}(E, (c_i, c^*_{-i})) = c^*_{-i}$ for every $i \in N$
and every $c_i \ge E - \sum c^*_{-i}$
(note that this condition is just a paraphrase of
``$z^w_{-i}(E, (c_i, c^*_{-i})) = c^*_{-i}$ for every $i \in N$
and every $c_i$ with $(E, (c_i, c^*_{-i})) \in \problems^N$'').
To see the ``only if'' direction,
given $i \in N$ and $c_i$ with
$z^w_{-i}(c_i, c^*_{-i}) \neq c^*_{-i}$,
whence $z^w(c_i, c^*_{-i}) \neq c^*$,
choose $\pi_i$ so that $\pi_i(z^w(c_i, c^*_{-i})) < \pi_i(c^*)$.
To see the other direction,
note that if $z^w_{-i}(E, (c_i, c^*_{-i})) = c^*_{-i}$ for all $i$
and $c_i\ge E - \sum c^*_{-i}$,
then
\[
  (\pi \circ z^{w, E})(c_i, c^*_{-i}) =
  \begin{cases}
    \pi(c_i, c^*_{-i}), & (c^*_{-i} \ge E - \sum c^*_{-i})\\
    \pi(c^*), & (c^*_{-i} < E - \sum c^*_{-i})
  \end{cases}
\]
which does not exceed $\pi(c^*)$ by assumption.

Let us recall that, under our current convention,
(weighted) axioms are indexed by weights, not vectors of amount of goods.
Also relevant is our assumption that weights are members of $\Delta = \Delta^{|N| - 1}$.
Our axiom, therefore, should not depend on the original Nash equilibrium $c^*$ itself,
but $c^* / \sum c^*$.
Given a weight~$w$,
it would then demand of a rule $z$ that $Q(\lambda w, z)$
for \emph{every} $\lambda \in \nonnegreal$.
We have now arrived a new definition:
\begin{defi}\label{defi:sf}
  A rule $z$ is \emph{strategy-free}$^w$ for a weight $w$
  if for every $i \in N$, $\lambda \in \nonnegreal$,   $E \ge \sum \lambda w$,
  and $c_i \ge E - \sum \lambda w_{-i}$,
  \begin{equation}
    z_{-i}(E, (c_i, \lambda w_{-i})) = \lambda w_{-i}.\label{eq:defi}
  \end{equation}
\end{defi}
\begin{rem*}
  \begin{enumerate}
  \item As before, the last part of the definition is equivalent to saying
``for every $c_i$ with $(E, (c_i, \lambda w_{-i})) \in \problems^N$, equation~(\ref{eq:defi}) obtains.''
\item Expecting Lemma~\ref{lem:composite},
  let a binary relation $D \subseteq \problems^N \times \Delta$
  consist of tuples $((E, c), w) \in \problems^N \times \Delta$ such that
\[
  \addtocounter{equation}{1}
    \aru*{i =: i(c, w)\in N}\aru*{c'_i =: c'_i(c, w)\stackrel{(\arabic{equation})}{<} c_i} \left[(c'_i, c_{-i}) \parallel w, 
    E \ge \sum c_{-i} + c'_i\right]
\]
note the strict inequality (\arabic{equation}).
Then, a rule $z$ is strategy-free$^w$ if and only if
$z_{-i}(E, c) = c_{-i}$ for every $(E, c) \in \problems^N$
with $(c, w) \in D$.
\end{enumerate}
\end{rem*}

In addition to requiring (\ref{eq:defi}) of a rule~$z$,
one might want to demand that (\ref{eq:defi}) hold approximately
for close arguments of $z$.
Corollary~\ref{cor:null} of \S~\ref{sec:main-results}
motivates such an additional requirement.
There are several ways in which to formalize this idea.
Expecting Proposition~\ref{prop:cont},
we propose the following axiom. 
\begin{defi}\label{defi:csf}
  A rule~$z$ is \emph{continuously strategy-free$^w$}
  if 
  for arbitrary $\epsilon > 0$
  there exists $\delta > 0$
  such that for every $E > 0$,
  $c \in \nonnegreal^N$,
  $i \in N$, and
  $c'_i \in \nonnegreal$,
  the conditions
  $E \ge \sum \lambda w$, $(E, (c_i', c_{-i})) \in \problems^N$, and
  $\lVert \lambda w_{-i} - c_{-i} \rVert < \delta$
  imply
  $\lVert z_{-i}(E, (c_i', c_{-i})) - c_{-i} \rVert < \epsilon$.
\end{defi}
Informally, this condition requires that
$z_{-i}(E, (c_i', c_{-i}))$ is close to $c_{-i}$
whenever there is $\lambda$ such that $\lambda w_{-i}$ is close enough to $c_{-i}$.

\cite{Cho2014} inspire another family of axioms.  First, we need a definition:
\begin{defi}
  For a rule $z$, the vector $\alpha(z) \in \nonnegreal^N$ is defined by
  \[
    \addtocounter{equation}{1}
    \alpha_i(z) = \inf \{ \rho \ge 0 \mid \zenbu{(E, c) \in \problems^N} {z_i(E, c) \stackrel{(\arabic{equation})}{\ge} c_i \wedge \rho E} \}.
  \]
\end{defi}
Note that (\arabic{equation}) is an equality 
whenever $c_i =  c_i \wedge \rho E$.

For each weight $w$, ``$\alpha(\cdot) = w$'' is a property of a rule.
It is worth noting that $\sum \alpha(z)$ may be strictly less than $1$, i.e.,
$\alpha(z)$ may not be a weight
as they are usually understood.
\cite{Cho2014} call such rules \emph{individually responsive},
and other rules \emph{individually unresponsive}.
Also note that if $w = (1/|N|, \dots, 1/|N|)$,
then this rule is nothing but what \mbox{\cite{Moulin2002}} introduced and
\cite{Thomson2019} called \emph{min-of-claim-and-equal division lower bounds on awards}.
\section{Main results}\label{sec:main-results}
The theory for the case $|N| = 2$ is quite different from the rest.
The following two Facts are  pointed out by, e.g., \mbox{\cite{Moulin2002}}
for the unweighted case, e.g., $w_i = w_j$ for all $i, j \in N$.
\begin{fact}[folklore]
  If $|N| = 2$, then for each weight $w \in \Delta$,
  the axiom $\alpha(\cdot)= w$ characterizes $\rulename{CEA}^w$.
\end{fact}

Therefore, we assume $|N| \ge 3$ for the entirety of this section.

\begin{fact}[folklore]
  Let $w$ be a weight.
  \begin{enumerate}
  \item $\alpha(\rulename{CEA}^w) = w$
  (hence $\rulename{CEA}^w$ is individually unresponsive).
\item 
  The axiom $\alpha(\cdot) = w$
  does \emph{not} characterize $\rulename{CEA}^w$.
\end{enumerate}
\end{fact}
This Fact be strengthened to the following,
which may well be known but whose proof the author was unable to find.
\begin{prop}\label{prop:non-char}
  Claims monotonicity and $\alpha(\cdot) = w$
  does not characterize $\rulename{CEA}^w$.
\end{prop}
\begin{proof}
  For simplicity, an argument for the case $N = \{0, 1, 2\}$ and
  $w = (1/3, 1/3, 1/3)$ will be presented.
  We construct the desired rule $z$ by the following.
  Let $(E, c) \in \problems^N$ be arbitrary.
  If $I(E, c) := \{i \in N \mid c_i < E/3\}$ is not a singleton, 
  then $\alpha(z) = w$ indeed determines $z(E, c)$.
  Otherwise, let $\{i\} = I(E, c)$.
  We still are forced to set $z_i(E, c) = c_i$.
  Let $j \in N$ be the least number
  such that $c_j = \max_{k \in N \setminus \{i\}} c_k$.
  Let $z_j(E, c) = (E - c_i) \wedge c_j$;
  the other claimant's award is determined by (\ref{eq:balance}).
\end{proof}

In light of the Proposition,
Theorem~2 of \cite{Cho2014} should be read as
merely establishing the equivalence between 
a special case of strategy-freedom$^w$ 
and $\alpha(\cdot) = w$
under the assumption of claims monotonicity\footnote{%
  Note that in \cite{Cho2014}, claims monotonicity is part of the definition of rules.},
as those conditions do not determine a rule uniquely.
In any case,
the following generalizes their Theorem~2.
\begin{thm}\label{thm:sf-alpha}
  Let $z$ be an individually unresponsive rule, i.e., $\sum \alpha(z) = 1$.
  It is strategy-free$^w$ for a weight $w$ if and only if
  $\alpha(z) = w$.
\end{thm}
\begin{proof}
  Suppose first $\alpha(z) = w$.
  Take an arbitrary $\lambda \in \nonnegreal$, $E \ge \sum \lambda w$, $i \in N$, and $c_i > E - \sum \lambda w_{-i}$.
  Then $\lambda w_{-i} \le \alpha_{-i}(z)E$,
  and, by the definition of $\alpha$,
  we have $z_{-i}(E, (c_i, \lambda w_{-i})) = \lambda w_{-i}$ as desired.

  On the other hand, if $\alpha(z) \neq w$,
  then since $\sum \alpha(z) = \sum w = 1$,
  there exists $i \in N$ such that $\alpha_i(z) > w_i$.
  Take a positive $\lambda$, and
  let $E = \sum \lambda w$ and $c_i = E\alpha_i(z)$, which is greater than $\lambda w_i$.
  Again by the definition of $\alpha$,
  we have $z_i(E, (c_i, \lambda w_{-i})) = c_i > \lambda w_i$.
  Since $\sum z(E, (c_i, \lambda w_{-i})) = E = \sum \lambda w$,
  it follows that $\sum z_{-i}(E, (c_i, \lambda w_{-i})) < \sum \lambda w_{-i}$,
  whence $z_{-i}(E, (c_i, \lambda w_{-i})) \neq \lambda w_{-i}$.
\end{proof}
\begin{cor}\label{cor:iu-exact}
  Let $w$ and $w'$ be weights.
  \begin{enumerate}
  \item If an individually unresponsive $z$ is 
     both strategy-free$^w$ and strategy-free$^{w'}$,
    then $w = w'$.
  \item If $\rulename{CEA}^w$ is strategy-free$^{w'}$, then $w = w'$.
  \end{enumerate}
\end{cor}

We now move on to individually responsive rules.
\begin{prop}\label{prop:myrule}
  There exists an  individually responsive strategy-free$^w$ rule.
\end{prop}
\begin{proof}
  Consider the following rule $z$ for $N = \{0, 1, 2\}$ and $w = (1/3, 1/3, 1/3)$.
  Given $(E, c) \in \problems^N$ with $c_0 \ge c_1 \ge c_2$, define
  \[
    z(E, c) =
    \begin{cases}
       ((E - 2c_2) \vee E/3, c_2 \wedge E/3, c_2 \wedge E/3) & \text{if $c_1 = c_2$},\\
       (c_1 \wedge E/2, c_1 \wedge E/2, (E - 2c_1) \vee 0) & \text{otherwise}.
    \end{cases}
  \]
  and extend the definition using anonymity.
\end{proof}
\begin{rem*}
  \begin{enumerate}
  \item The rule above is claims monotonic as well.
    Hence it satisfies the conditions Corollary~1 of \cite{Cho2014},
    despite what is claimed at the end of their \S~4.
  \item The Proposition alternatively follows from Corollary~\ref{cor:AC}.
    The rule defined above has the advantages of being explicit
    and claims monotonic.
  \end{enumerate}
\end{rem*}

Corollary~\ref{cor:iu-exact} leaves unanswered the possibility of
strategy-freedom$^w$ of an individually \emph{responsive} rule
for different $w$'s at once.
We give an affirmative, albeit non-constructive, answer in the Appendix
(Corollary~\ref{cor:AC}).
We instead present a related negative result here.
Going back to the original model,
even if we have a good reason to believe the original Nash equilibrium
is $w$ after normalization,
we may have committed errors in estimation and measurement.
It is therefore truly desirable for a rule~$z$ to be strategy-free$^{w'}$
for all $w'$ \emph{close enough} to $w$.
Mathematically,
this is about existence of an \emph{open neighborhood}~$U$
of $w$ such that $z$ is strategy-free$^{w'}$ for all $w' \in U$.
Unfortunately, this is ruled out by the corollary below.


In the following, let us call pairs $(u^i, u^j)$
of distinct weights \emph{bad}
(or $(i,j)$-bad, if we want to be explicit on the indices)
if  there are distinct $i, j \in N$
such 
that $u^i_{-i-j} \parallel u^j_{-i-j}$,
that $u^i_i\sum u^j_{-i-j} < u^j_i \sum u^i_{-i-j}$
and that $u^i_j\sum u^j_{-i-j} > u^j_j\sum u^i_{-i-j}$.

\begin{lem}\label{lem:compositehalf}
  Let $(u^i, u^j)$ be an $(i, j)$-bad pair of weights.
  There is no rule that is both strategy-free$^{u^i}$ and strategy-free$^{u^j}$.
\end{lem}
\begin{proof}
  Suppose that $z$ is strategy-free$^{u^k}$  for each $k \in \{i, j\}$.
  Take some $c_{-i-j} \in \nonnegreal^{N - \{i, j\}}$
  parallel to $u^i_{-i-j}$ (and thus to $u^j_{-i-j}$).
  Then,
  for each of $k \in \{i, j\}$,
  every $E$
  at least
  $L_k := \sum c_{-k} + u^k_k (\sum c^k_{-k})/(\sum u^k_{-k})$ and
  at most $\sum c$,
  and $c_i, c_j \in \nonnegreal$,
  strategy-freedom$^{u^k}$ implies that
  $z_{-k}(E, c) = c_{-k}$.
  Let $c_k = u^{k'}_k (\sum c^k_{-k})/(\sum u^k_{-k})$ for $(k, k') = (i, j), (j, i)$.
  By badness, $L_k < \sum c$ ($k = i, j$).
  Let $E = L_i \vee L_j$;
  without loss of generality, assume $L_i \ge L_j$ and thus $E = L_i$.
  Then by strategy-freedom$^{u_i}$,
  \[
    z_i(E, c) = E - \sum z_{-i}(E, c) = E - \sum c_{-i} = u^i_i(\sum c^k_{-k})/(\sum u^k_{-k}).
  \]
  By strategy-freedom$^{u_j}$,
  \[
    z_i(E, c) = (z_{-j}(E,c))_i = c_i.
  \]
  These cannot all be equal due to badness.
\end{proof}

\begin{cor}\label{cor:null}
  Suppose that a rule $z$ is strategy-free$^w$ for all $w \in S$,
  where $S \subseteq \Delta$.
  Then $S$ is a null set.
\end{cor}
\begin{proof}
  Choose distinct $i \neq j$ in $N$.
  For $p \in \nonnegreal^{N - \{i, j\}}$,
  the line $l_p := \{w \in \Delta \mid w_{-i-j} = p\}$
  intersects $S$ with at most one point.
  Indeed, if $u, v \in l_p \cap S$ are distinct,
  either $(u, v)$ or $(v, u)$ is an $(i,j)$-bad pair.
  The claim then follows from a measure-theoretic argument.
\end{proof}

By Corollary~\ref{cor:null},
for any weight $w$,
there is no open neighborhood $U$ of $w$ a rule $z$ such that
$z$ is strategy-free$^{w'}$ for all $w' \in U$.
The next best thing to hope for, therefore, is something like
continuous strategy-freedom$^w$.
The following fact suggests that this particular formalization is a natural axiom.
\begin{prop}\label{prop:cont}
  A rule is  continuously strategy-free$^w$
  if and only if it is strategy-free$^w$ and uniformly claims continuous
  on $\problems^N$.
\end{prop}
\begin{proof}
  We prove the ``if'' direction first.
  Let $\epsilon > 0$ be given.
  We may take $\delta_{\text{uc}}>0$ witnessing the
  uniform continuity of $z$ for $\epsilon$
  and let $\delta = \delta_{\text{uc}} \wedge \epsilon$.
  Let $E, \lambda, i, c, c'_i$ be arbitrary as in Definition~\ref{defi:csf},
  so that $\lVert (c'_i, c_{-i}) - (c'_i, \lambda w_{-i}) \rVert
  = \lVert c_{-i} - \lambda w_{-i}\rVert < \delta$.
  By strategy-freedom$^w$, we have $z_{-i}(E, (c_i', \lambda w_{-i})) = \lambda w_{-i}$.
  Therefore,
  \begin{align*}
  \lVert z_{-i}(E, (c'_i, c_{-i})) - c_{-i} \rVert
  &\le
    \begin{aligned}[t]
&\lVert z_{-i}(E, (c'_i, c_{-i})) - z_{-i}(E, (c'_i, \lambda w_{-i}))\rVert\\
  &+ \lVert \lambda w_{-i} - c_i \rVert
\end{aligned}
\\
  &\le 2\epsilon
  \end{align*}as desired.

  The ``only if'' part can be proved as follows.
  The strategy-freedom$^w$ is clear.
  Let $\epsilon > 0$ be arbitrary.
  Take $\delta_{\text{sf}}$ witnessing the continuous strategy-freedom$^w$
  of $z$ for $\epsilon$.
  Define $\delta = \delta_{\text{sf}} \wedge \epsilon$.
  Consider arbitrary $c, d \in \nonnegreal^N$ with $\lVert c - d \rVert < \delta$.
  Take any $i \in N$.
  There exist $\lambda, \mu > 0$
  such that \[\lVert c - d \rVert \ge (\lVert c_{-i} - \lambda w_{-i} \rVert
  + \lVert d_{-i} - \mu w_{-i} \rVert) \vee \lVert \lambda w_i - \mu w_i\rVert.\]
  Note that $z_i(E, x) = E - \sum z_{-i}(E, x)$ for $x = c, d$,
  whence $|z_i(E, c) - z_i(E, d)| \le \lVert z_{-i}(E, c) - z_{-i}(E, d) \rVert_1$.
  Therefore,
  \begin{align*}
    \lVert z(E, c) - z(E, d) \rVert
    &\le |z_i(E, c) - z_i(E, d)| + \lVert z_{-i}(E, c) - z_{-i}(E, d) \rVert \\
    &\le \left(\frac{1}{\sqrt 2} + 1\right) \lVert z_{-i}(E, c) - z_{-i}(E, d) \rVert \\
    &\le
      \left(\frac{1}{\sqrt 2} + 1\right)
      \begin{aligned}[t]
    (&\lVert z_i(E, c) - z_i(E, (c_i, \lambda w_i))  \rVert\\
      &\quad + \lVert \lambda w_i - \mu w_i \rVert\\
      &\quad + \lVert z_i(E, (d_i, \mu w_i)) - z_i(E, d) \rVert)
    \end{aligned}
\\
    &\le 3\left(\frac{1}{\sqrt 2} + 1\right)\epsilon,
  \end{align*}
  as desired.
\end{proof}

We conclude the section by a characterization  of $\rulename{CEA}^w$
in terms of strategy-freedom$^w$ for
general $N$.
It is instructive to present our result as a corollary of
a characterization of the following \emph{class}
of rules 
containing $\rulename{CEA}^w$.
\begin{defi}
  For a $\kappa \in \nonnegreal$ and a weight $w \in \Delta$,
  the rule $\rulename{CEA}^w_\kappa$ is defined by
  \[
    \rulename{CEA}^w_\kappa(E, c) = (c_i \wedge w_i (c_i/w_i)^{-\kappa} \lambda)_{i \in N},
  \]
  where $\lambda$ is chosen to satisfy (\ref{eq:balance}),
  for every $(E, c) \in \problems^N$.
\end{defi}
Obviously, $\rulename{CEA}^w = \rulename{CEA}_0^w$.

\begin{thm}
  For each $w \in \Delta$,
  the rules of the form $\rulename{CEA}^w_\kappa$ are
  the exactly separable directional ones
  that are strategy-free$^w$ and homogeneous.
\end{thm}
\begin{proof}
  It is easy to see that $\rulename{CEA}_\kappa^w$ is strategy-free$^w$ and homogeneous.
  Moreover, $\rulename{CEA}^w_\kappa = P^u$,
  where $u_i(c) = w_i(c_i/w_i)^{-\kappa}$.

  Conversely, assume that $P^u$ is, strategy-free$^w$, and homogeneous.
  Consider an arbitrary tuple~$((E, (c'_i, \lambda w_{-i})), w) \in D$.
  By strategy-freedom$^w$,
  for every $j \in N \setminus \{i\}$,
  we have
  \begin{equation}
    u_j(\lambda w_j) / u_i(c'_i) \ge \lambda w_j / c'_i.\label{eq:sf-u}
  \end{equation}
  Letting $c'_i = \lambda w_i$, (does this require endowment continuity?)
  we obtain $u_j(\lambda w_j) / u_i(\lambda w_i) \ge w_j / w_i$.
  Since $i \neq j$ are arbitrary,
  $u(j, \lambda w_j) / u(i, \lambda w_i) = w_j / w_i$.
  Thus, there exists a function $f: \posreal \to \posreal$ such that
  \begin{equation}
    u_i(\lambda w_i) = w_i f(\lambda)\label{eq:defif}
  \end{equation}
  for every $i \in N$ and $\lambda > 0$,
  i.e.,
  $u_i(c_i) = w_i f(c_i / w_i)$ for every $c_i \in \nonnegreal$ and $i \in N$.

  Going back to (\ref{eq:sf-u}),
  letting $c'_i = \lambda' w_i$ for $\lambda' \ge \lambda$ this time,
  and using the definition~(\ref{eq:defif}),
  we have $f(\lambda') \le f(\lambda)$.
  Thus $f$ is non-increasing.

  Now, one can show that homogeneity implies $u_i(c) \parallel u_j(\alpha c)$
  for all $c \ge 0$ and $\alpha > 0$.
  In other words,
  $u_i(\alpha \lambda w_i)/u_i(\lambda w_i) = u_j(\alpha \lambda' w_j) / u_j(\lambda' w_j)$
  for all $\alpha, \lambda, \lambda' > 0$ and $i, j \in N$.
  By (\ref{eq:defif}), this implies, in our case,
  \begin{equation}
    \frac{f(\alpha\lambda)}{f(\lambda)} = \frac{f(\alpha\lambda')}{f(\lambda')}\label{eq:mult}
  \end{equation}
  for arbitrary $\alpha, \lambda, \lambda' > 0$.
  Define a function $F: \mathbb R \to \mathbb{R}$ by
  $F(x) = \log f (\exp x)$,
  where the common base is $f(1)$, which we shall assume henceforth.
  Since $f$ is non-increasing, so is $F$.
  Now for $x, y \in \mathbb R$,
  \begin{align}
    F(x) + F(y) - F(xy) &= \log\left( \frac{f(\exp x)f(\exp y)}{f(\exp(x+y))} \right)\notag\\
                        &= \log f(1) \label{eq:logf1}\\
                        &= 0\notag,
  \end{align}
  where (\ref{eq:mult}) with
  $(\alpha, \lambda, \lambda') = (\exp x, 1, \exp y)$
  is used to derive (\ref{eq:logf1}).
  Since $F$ is non-increasing,
  $F(x) = -\kappa x$ for some $\kappa \in \nonnegreal$.
  Unraveling the relevant definitions leads to the desired conclusion.
\end{proof}

\begin{cor}\label{cor:char}
  For each $w \in \Delta$,
  the rule $\rulename{CEA}^w$ is the only separable directional rule
  that is strategy-free$^w$, homogeneous, and claims monotonic.
\end{cor}
\begin{proof}
  In light of the Theorem,
  it suffices to show that $\rulename{CEA}^w_\kappa$ is claims monotonic
  if and only if $\kappa = 0$.
  Obviously, $\rulename{CEA}^w_0$ is claims monotonic.
  On the other hand, assume that $\kappa > 0$ and,
  for simplicity, that $N = \{0, 1, 2\}$.
  Consider $c \in \nonnegreal^N$ such that $c_0/w_0 < c_1 / w_1 = c_2 / w_2$
  and $E  \in (c_0, c_0 + c_1 + c_2)$.
  Here, one can see that $(\rulename{CEA}^w_\kappa(E, c))_1 = (\rulename{CEA}^w_\kappa(E, c))_2$.
  Now take $c'_2 > c_2$, so that $c_1 / w_1 < c_2 / w_2$.
  We then have $(\rulename{CEA}^w_\kappa(E, (c_2', c_{-2})))_1 < (\rulename{CEA}^w_\kappa(E, c_2', c_{-2}))_2$.
  On the other hand, $(\rulename{CEA}^w_\kappa(E, c))_0 = (\rulename{CEA}^w_\kappa(E, (c_2', c_{-2})))_0$,
  whence by (\ref{eq:balance}),
  $(\rulename{CEA}^w_\kappa(E, c))_2 > (\rulename{CEA}^w_\kappa(E, (c_2', c_{-2})))_2$,
  which violates claims monotonicity.
\end{proof}

The preceding Theorem and Corollary can be stated in an axiomatic manner 
by quoting the axiomatization of separable directional rules by \mbox{\cite{Harless2016}}.

\section{Concluding remarks}\label{sec:concluding-remarks}
In the present article,
we presented a new axiom, strategy-freedom, on rules for claims problems.
The axiom is inspired by, and has potential applications in,
supply chain research.
If goods whose amount is limited are distributed to retailers according to a
strategy-free$^w$ rule,
the retailers in a competition whose Nash equilibrium is $c^*$ have no incentive to change their orders from $c^*$ as long as it is a multiple of $w$.
The relevance of the axiom is further confirmed by the fact
(Theorem~\ref{thm:sf-alpha})
that, under certain contidions,
it is equivalent to a natural generalization of
a known axiom (min-of-claim-and-equal division lower bounds on awards)
and the fact (Corollary~\ref{cor:char}) that it
characterizes $\rulename{CEA}^w$,
the weighted constrained equal awards rule,
together with other natural conditions.

It remains an important future task
to analyze game-theoretically a variant of strategy-freedom we introduced,
continuous strategy-freedom$^w$,
in the original context of retail competition
and evaluate the degree, if any, to which a rule satisfying that axiom
might incentivize the retailers to inflate their orders
in case their original Nash equilibrium is not a multiple of $w$.

\section*{Acknowledgements}
The author has been financially supported by the Corporate
  Sponsored Research Program in  Advanced Logistic Science
  at the Research Center for Advanced Science and Technology,
  The University of Tokyo.
  The members of the Research Group, as well as Reid Dale and Nariaki Nishino,
  have provided useful feedback.

\appendix

\section{A rule strategy-free$^w$ for infinitely many $w$}\label{sec:rule-strategy-freew}
If we allow a pathological $W \subseteq \Delta$,
we can have an infinite $W$ and a single rule that is strategy-free$^w$ for all $w \in W$
at once.
The notion of bad pairs, which we used to derive the negative result earlier,
turns out to be relevant.

\begin{prop}\label{lem:composite}
  Let $B$ be the binary relation on $\Delta$ consisting of bad pairs.
  Recall the definition of the binary relation $D \subseteq \problems^N \times \Delta$
  right after Definition~\ref{defi:sf}.
  Define $B' \subseteq \Delta \times \Delta$
  as
  $(w,  w') \in B'$ if and only if there exists $i \in N$
  with $w_{-i} \parallel w'_{-i}$.
  The union $B \cup B'$ is the composite of $D$ and $D^\inv$.
\end{prop}
\begin{proof}
  That $B$ is included in the composite of $ D$ and $D^\inv$ follows from the proof of Lemma~\ref{lem:compositehalf}.
  It is easy to see that $B' \subseteq D \circ D^\inv$. 
  Conversely,
  suppose that $u^i \mathbin{D^\inv} (E, c) \mathbin{D} u^j$,
  where $u^i, u^j \in \Delta$ and $(E, c) \in \problems^N$.
  Suppose first that $i(u^i, c) = i(u^j, c) =: i$,
  where $i(\cdot, \cdot)$ is as in the definition of $D$.
  Then,
  $u^i_{-i} \parallel c_{-i} \parallel u^j_{-i}$
  by the definition of $D$, whence $(u^i, u^j) \in B'$.
  Secondly, assume that $i(u^i, c)$ and $i(u^j, c)$ are distinct,
  which we let $i$ and $j$, respectively.
  We show that $(u^i, u^j)$ are bad.
  As before, $u^i_{-i} \parallel c_{-i}$ and $c_{-j}\parallel u^j_{-j}$,
  whence $u^i_{-i-j} \parallel u^j_{-i-j}$.
  Furthermore,
  \begin{equation}
    c'_k(c, u^k) < c_k\label{eq:ineq}
  \end{equation}
for both $k = i, j$,
  where the left-hand side is also as in the definition of $D$.
  Because $(c'_k(c,u^k), c_{-k}) \parallel u^k$ ($k = i, j$),
  we have
  \begin{equation}
    c'_k(c, u^k) = u^k_k (\sum c_{-i-j})/ (\sum u_{-i-j})\label{eq:proportionbyhead}
  \end{equation}
  and
  \begin{equation}
    c_k = u^{k'}_k(\sum c_{-i-j})/ (\sum u_{-i-j})\label{eq:proportionbytail},
  \end{equation}
  where $(k, k') = (i, j), (j, i)$.
  Combining equations~(\ref{eq:ineq}--\ref{eq:proportionbytail}),
  we obtain the defining inequalities for the badness of $(u^i, u^j)$.
\end{proof}

\begin{cor}\label{cor:AC}
  There exist an infinite set $W \subseteq \Delta$
  and a rule $z$ such that
  $z$ is strategy-free$^w$ for all $w \in W$.
\end{cor}
\begin{proof}
  By the Axiom of Choice, there exists a maximal $(B \cup B')$-independent
  $W \subseteq \Delta$.
  That is,
  for every distinct pair $\{w, w'\} \subseteq W$,
  exactly one of the weights is in $W$ if and only if they do not form a bad
  or $B'$-related pair.
  By construction, for each $(E, c) \in \problems^N$,
  there exists \emph{exactly one}
  $w \in W$, for which we write $w(c)$,
  such that $((E, c), w) \in D$.
  We may define a rule $z$ so
  for each $(E, c) \in \problems^N$,
  the award $z(E, c)$ is the unique vector that is  parallel to $w(c)$
  and is identical to $c$
  possibly except in one coordinate.
  It is evident that $z$ is strategy-free for all $w \in W$ by construction.
\end{proof}


\bibliographystyle{elsarticle-harv} 
\bibliography{econ}


\end{document}